\documentclass[11pt]{article}
\usepackage{geometry}
\usepackage{amsmath, amssymb, amsthm}
\usepackage{graphicx}
\usepackage{bm}
\usepackage[linesnumbered,ruled]{algorithm2e}
\newtheorem{theorem}{Theorem}
\newtheorem{lemma}{Lemma}

\newtheorem{claim}{Claim}
\title{A Matroid Generalization of the Super-Stable Matching Problem}
\author{Naoyuki Kamiyama%
\thanks{Institute of Mathematics for Industry, Kyushu University, Fukuoka, Japan.}
\thanks{JST, PRESTO, Kawaguchi, Japan.}
\thanks{{\ttfamily kamiyama@imi.kyushu-u.ac.jp}}
}
\date{}

\begin{document}

\maketitle

\begin{abstract}
A super-stable matching, which was 
introduced by Irving, is a solution concept in a variant of 
the stable matching problem in which the preferences may contain ties. 
Irving proposed a polynomial-time algorithm for the problem of 
finding a super-stable matching if a super-stable matching exists. 
In this paper, we consider a matroid generalization of a super-stable matching. 
We call our generalization of a super-stable matching 
a super-stable common independent set. 
This can be considered as a generalization of the matroid generalization of 
a stable matching for strict preferences proposed by Fleiner. 
We propose a polynomial-time algorithm for the problem of 
finding a super-stable common independent set 
if a super-stable common independent set exists.
\end{abstract} 

\section{Introduction}

The stable matching problem, which was introduced by 
Gale and Shapley~\cite{Gs62}, is one of the most central topics
in the study of matching under preferences. 
The aim of this paper is to unify the following two directions of generalization 
of the stable matching problem. 

The first direction is generalization of preferences. 
More concretely, in this generalization, we allow 
the preferences contain ties. 
That is, some agent may be indifferent between some potential partners. 
A super-stable matching, which was 
introduced by Irving~\cite{I94}, is a solution concept in this generalization.
This stability concept guarantees that there does not exist 
an unmatched pair $\{a,b\}$ such that $a$ (resp., $b$) 
prefers $b$ (resp., $a$) to the current partner, 
or is indifferent between $b$ (resp., $a$) and the current partner.
Irving~\cite{I94} proposed a polynomial-time algorithm for the problem of 
finding a super-stable matching if a super-stable matching exists. 

The second direction is generalization of matchings.
In the stable matching problem, we usually consider a matching in 
a bipartite graph. 
In this generalization, we generalize a matching in a bipartite graph to a 
common independent set of two matroids. 
Fleiner~\cite{F03} introduced a matroid generalization of 
the stable matching problem with strict preferences, and 
proved that 
there always exists a stable solution in this generalization and 
we can find a stable solution in polynomial time if we are given 
polynomial-time independence oracles for the matroids. 

In this paper, we consider a matroid generalization of a super-stable matching. 
We call this generalization of a super-stable matching 
a super-stable common independent set, and we 
consider the problem of finding a super-stable common 
independent set if a super-stable common 
independent set exists. 
The contribution of this paper is a polynomial-time algorithm for 
this problem.
Although special cases of this problem were considered in \cite{K15,K19}, 
the polynomial-time solvability of this problem
has been open.
We affirmatively settle this question. 
Our algorithm can be considered as a generalization of 
the algorithm proposed in \cite{OM20+}
for finding a super-stable matching in 
the student-project allocation problem with ties, 
which is a special case of our problem 
(see Appendix~\ref{appendix:spa}).

\subsection{Related work} 

Irving~\cite{I94} introduced a super-stable matching 
in the stable matching problem with ties
(see, e.g., \cite{IM08} and \cite[Chapter~3]{M13} 
for a survey of the stable matching problem with ties). 
One of the most notable differences between a super-stable matching 
and a stable matching in the 
stable matching problem with strict preferences
is that there does not necessarily exist a super-stable matching~\cite{I94}. 
In the one-to-one setting, 
Irving~\cite{I94} 
proposed a polynomial-time algorithm for 
finding a super-stable matching 
if a super-stable matching exists
(see also \cite{M99}). 
In the many-to-one setting, 
Irving, Manlove, and Scott~\cite{IMS00} proposed 
a polynomial-time algorithm for 
finding a super-stable matching
if a super-stable matching exists.
Scott~\cite{S05} considered 
a super-stable matching in the many-to-many setting.
Furthermore, 
Olaosebikan and Manlove~\cite{OM20+}
considered a super-stable matching in the 
student-project allocation problem with ties. 
As a special case, the situation in which we are given a master list and 
the preference lists of agents are 
derived from this master list has been considered. 
In the one-to-one setting with master lists, 
Irving, Manlove, and Scott~\cite{IMS08} gave 
a simple polynomial-time algorithm for finding a 
super-stable matching
if a super-stable matching exists. 
In the many-to-one setting with master lists, 
O'Malley~\cite{O07} gave a
polynomial-time algorithm for 
finding a super-stable matching 
if a super-stable matching exists.

Fleiner~\cite{F03} considered a matroid generalization of the stable matching 
problem with strict preferences. 
Matroid approaches to the stable matching problem have been extensively 
studied (see, e.g., \cite{FK16,FT07,IY16,K15,KTY18,MY15,Y17}). 
Special cases of the problem considered in this paper were considered in \cite{K15,K19}.
More concretely, in \cite{K15,K19}, the situation in which 
we are given a master list was considered. 
Furthermore, a related problem was considered in \cite{K20}. 

\section{Preliminaries} 

For each set $X$ and each element $u$, we define 
$X + u := X \cup \{u\}$ and 
$X - u := X \setminus \{u\}$. 

An ordered pair ${\bf M} = (U,\mathcal{I})$ of a finite set $U$ and 
a non-empty family $\mathcal{I}$ of subsets of $U$ is called a 
{\em matroid} if
for every pair of subsets $I,J \subseteq U$,  
the following conditions are satisfied.
\begin{description}
\setlength{\parskip}{0mm} 
\setlength{\itemsep}{1mm} 
\item[(I1)]
If $J \in \mathcal{I}$ and $I \subseteq J$, then $I \in \mathcal{I}$.
\item[(I2)]
If $I,J \in \mathcal{I}$ and $|I| < |J|$, then there exists an element $u \in J \setminus I$ 
such that $I + u \in \mathcal{I}$.
\end{description}

Assume that we are given a matroid ${\bf M} = (U, \mathcal{I})$. 
A member $I \in \mathcal{I}$ is called an {\em independent set} of ${\bf M}$. 
A subset $B \subseteq U$ is called a {\em base} if 
$B$ is an inclusion-wise maximal independent set of ${\bf M}$. 
Notice that the condition (I2) implies that 
all the bases of ${\bf M}$ have the same size. 
A subset $I \subseteq U$ such that $I \notin \mathcal{I}$ is called a 
{\em dependent set} of ${\bf M}$. 
A subset $C \subseteq U$ is called a {\em circuit} of ${\bf M}$ if 
$C$ is an inclusion-wise minimal dependent set of ${\bf M}$. 
Notice that the definition of a circuit implies that 
for every pair of distinct circuits $C_1,C_2$ of ${\bf M}$, 
$C_1 \setminus C_2 \neq \emptyset$ and 
$C_2 \setminus C_1 \neq \emptyset$.  

\begin{lemma}[See, e.g., {\cite[Lemma~1.1.3]{O11}}] \label{lemma:elimination}
Assume that we are given a matroid ${\bf M}=(U,\mathcal{I})$ and 
distinct 
circuits $C_1,C_2$ of ${\bf M}$ such that 
$C_1 \cap C_2 \neq \emptyset$. 
Then for every element $u \in C_1 \cap C_2$,
there exists a circuit $C$ of ${\bf M}$ 
such that $C \subseteq (C_1 \cup C_2) - u$. 
\end{lemma}

Assume that we are given a matroid ${\bf M}=(U,\mathcal{I})$ and 
an independent set $I$ of ${\bf M}$.
Define 
\begin{equation*}
{\sf span}_{{\bf M}}(I) := 
\{u \in U \setminus I \mid I + u \notin \mathcal{I}\}.
\end{equation*} 
It is not difficult to see that 
for every element $u \in {\sf span}_{{\bf M}}(I)$, 
$I + u$ contains a circuit of ${\bf M}$ as  
a subset, and (I1) implies that $u$ belongs to this circuit.
Furthermore, Lemma~\ref{lemma:elimination}
implies that  such a circuit is uniquely determined. 
We call this circuit the {\it fundamental circuit} of $u$ 
for $I$ and ${\bf M}$, and we denote by 
${\sf C}_{{\bf M}}(u,I)$ this circuit.
It is well known (see, e.g., \cite[p.20, Exercise 5]{O11}) that 
for every element $u \in {\sf span}_{{\bf M}}(I)$, 
\begin{equation*}
{\sf C}_{{\bf M}}(u,I) 
= \{u^{\prime} \in I  + u \mid I + u - u^{\prime} \in \mathcal{I}\}.
\end{equation*} 
For each element $u \in {\sf span}_{{\bf M}}(I)$, 
we define ${\sf D}_{{\bf M}}(u,I) := {\sf C}_{{\bf M}}(u,I) - u$. 

\subsection{Problem formulation}

In this paper, we are given matroids ${\bf M}_{\rm D} = (E, \mathcal{I}_{\rm D})$ and 
${\bf M}_{\rm H} = (E, \mathcal{I}_{\rm H})$ such that 
for every element $e \in E$, $\{e\} \in \mathcal{I}_{\rm D} \cap \mathcal{I}_{\rm H}$. 
Furthermore, throughout this paper, we are given  
transitive and complete\footnote{%
For every symbol ${\rm S} \in \{{\rm D}, {\rm H}\}$ and 
every pair of elements $e,f \in E$, at least one of $e \succsim_{\rm S} f$ and 
$f \succsim_{\rm S} e$ holds.} binary relations 
$\succsim_{\rm D}$ and $\succsim_{\rm H}$ on $E$. 
A subset $I \subseteq E$ is called a {\em common independent set} of 
${\bf M}_{\rm D}$ and ${\bf M}_{\rm H}$ if
$I \in \mathcal{I}_{\rm D} \cap \mathcal{I}_{\rm H}$. 
For each symbol ${\rm S} \in \{{\rm D},{\rm H}\}$ and 
each pair of elements $e,f \in E$, 
if $e \succsim_{\rm S} f$ and 
$f \not\succsim_{\rm S} e$, then 
we write $e \succ_{\rm S} f$. 

Assume that we are given a common independent set $I$ of 
${\bf M}_{\rm D}$ and ${\bf M}_{\rm H}$. 
Then for each symbol ${\rm S} \in \{{\rm D},{\rm H}\}$, 
we define ${\bf dom}_{\rm S}(I)$ by 
\begin{equation*}
{\bf dom}_{\rm S}(I) := 
\{e \in {\sf span}_{{\bf M}_{\rm S}}(I) \mid 
\mbox{$f \succ_{\rm S} e$ for every element $f \in {\sf D}_{{\bf M}_{\rm S}}(e,I)$}\}. 
\end{equation*}
Then $I$ is said to be {\em super-stable} if 
\begin{equation*}
E \setminus I = {\bf dom}_{\rm D}(I) \cup {\bf dom}_{\rm H}(I).
\end{equation*} 
Then the goal of the {\em super-stable common independent set problem} 
is to determine 
whether there exists a super-stable common independent set of 
${\bf M}_{\rm D}$ and ${\bf M}_{\rm H}$.
Furthermore, if there exists a super-stable common independent set of 
${\bf M}_{\rm D}$ and ${\bf M}_{\rm H}$, then 
we find a super-stable common independent set of 
${\bf M}_{\rm D}$ and ${\bf M}_{\rm H}$. 

In this paper, 
we assume that for every symbol ${\rm S} \in \{{\rm D}, {\rm H}\}$ and 
every subset $I \subseteq E$, 
we can determine 
whether $I \in \mathcal{I}_{\rm S}$ in time bounded by a polynomial in $|E|$. 

\subsection{Basics of matroids} 

Throughout this subsection, we assume that 
we are given a matroid ${\bf M} = (U, \mathcal{I})$. 

The following lemma easily follows from 
Lemma~\ref{lemma:elimination}. 
For completeness, we give its proof. 

\begin{lemma} \label{lemma:circuit_subset}
Assume that we are given independent sets $I,J$ of 
${\bf M}$ such that $I \subseteq J$, and 
an element $u \in U \setminus J$.
Then if $u \in {\sf span}_{{\bf M}}(I)$, then 
$u \in {\sf span}_{{\bf M}}(J)$ and 
${\sf C}_{{\bf M}}(u,I) = {\sf C}_{{\bf M}}(u,J)$. 
\end{lemma}
\begin{proof}
Since ${\sf C}_{{\bf M}}(u,I) \subseteq J + u$, 
$u \in {\sf span}_{{\bf M}}(J)$. 
Define $C_I := {\sf C}_{{\bf M}}(u,I)$ and 
$C_J := {\sf C}_{{\bf M}}(u,J)$. 

Assume that 
$C_I \neq C_J$. 
Then since $u \in C_I \cap C_J$, 
Lemma~\ref{lemma:elimination} implies that 
there exists a circuit $C$ of ${\bf M}$ such that 
$C \subseteq (C_I \cup C_J) - u \subseteq J$.
However, this contradict the fact that 
$J \in \mathcal{I}$. 
\end{proof} 

The following stronger version of Lemma~\ref{lemma:elimination} is known. 

\begin{lemma}[See, e.g., {\cite[p.~15, Exercise~14]{O11}}] \label{lemma:strong_elimination}
Assume that we are given distinct 
circuits $C_1,C_2$ of ${\bf M}$ such that 
$C_1 \cap C_2 \neq \emptyset$. 
Then for every element $u \in C_1 \cap C_2$ and 
every element $w \in C_1 \setminus C_2$, 
there exists a circuit $C$ of ${\bf M}$ 
such that $w \in C \subseteq (C_1 \cup C_2) - u$. 
\end{lemma}

The following lemma easily follows from 
Lemma~\ref{lemma:strong_elimination}. 
For completeness, we give its proof. 

\begin{lemma} \label{lemma:circuit_union}
Assume that we are given 
circuits $C, C_1,C_2,\ldots,C_{\ell}$ of ${\bf M}$.
Furthermore, we are given 
distinct elements $u_1,u_2,\ldots,u_{\ell}, w \in U$ satisfying the following 
conditions. 
\begin{description}
\setlength{\parskip}{0mm} 
\setlength{\itemsep}{1mm} 
\item[(U1)]
$u_i \in C \cap C_i$ holds 
for every integer $i \in \{1,2,\ldots,{\ell}\}$. 
\item[(U2)]
$u_{i} \notin C_{j}$ holds for every 
pair of distinct integers $i,j \in \{1,2,\ldots,{\ell}\}$. 
\item[(U3)]
$w \in C \setminus (C_1 \cup C_2 \cup \cdots \cup C_{\ell})$.  
\end{description}
Then there exists a circuit $C^{\prime}$ of ${\bf M}$ 
such that 
\begin{equation*}
w \in C^{\prime} \subseteq \Big(C \cup \Big(\bigcup_{i=1}^{\ell}C_i\Big)\Big) 
\setminus \{u_1,u_2,\ldots,u_{\ell}\}.
\end{equation*} 
\end{lemma}
\begin{proof}
We prove that for every integer $x \in \{0,1,\ldots,\ell\}$,  
\begin{equation} \label{eq1:lemma:circuit_union}
\mbox{there exists a circuit $C^{\prime}$ of ${\bf M}$ 
such that } w \in C^{\prime} \subseteq \Big(C \cup \Big(\bigcup_{i=1}^{x}C_i\Big)\Big) 
\setminus \{u_1,u_2,\ldots,u_{x}\}
\end{equation}
by induction on $x$. 
If $x = 0$, then since we can take $C$ as $C^{\prime}$, 
\eqref{eq1:lemma:circuit_union} holds. 

Let $z$ be an integer in $\{1,2,\ldots,\ell\}$. 
Assume that \eqref{eq1:lemma:circuit_union} holds
when $x = z -1$. 
Let $C^{\prime}$ be a circuit of ${\bf M}$ 
satisfying \eqref{eq1:lemma:circuit_union} when $x = z -1$. 
If $u_z \notin C^{\prime}$, then 
$C^{\prime}$ satisfies \eqref{eq1:lemma:circuit_union} when $x = z$. 
Thus, we assume that $u_z \in C^{\prime}$. 
This implies that $u_z \in C_z \cap C^{\prime}$. 
Furthermore, $w \in C^{\prime} \setminus C_z$. 
Thus, Lemma~\ref{lemma:strong_elimination} implies that 
there exists a circuit $C^{\circ}$ of ${\bf M}$ such that 
$w \in C^{\circ} \subseteq (C^{\prime} \cup C_z) - u_z$. 
Furthermore, (U2) implies that 
$u_1,u_2,\ldots,u_{z-1} \notin C^{\circ}$. 
These imply that 
$C^{\circ}$ satisfies \eqref{eq1:lemma:circuit_union} when 
$x = z$. 
\end{proof} 

Assume that we are given a subset $X \subseteq U$.
Define 
\begin{equation*}
\mathcal{I}|X := \{I \subseteq X \mid I \in \mathcal{I}\}, \ \ \ 
{\bf M}|X := (X, \mathcal{I}|X).
\end{equation*} 
Then it is known (see, e.g., \cite[p.20]{O11}) that 
${\bf M}|X$ is a matroid. 
Define ${\sf r}_{{\bf M}}(X)$ as the size of a base of ${\bf M}|X$. 
For each subset $I \subseteq U \setminus X$, 
we define 
\begin{equation*}
{\sf p}_{{\bf M}}(I; X) := {\sf r}_{{\bf M}}(I \cup X) - {\sf r}_{{\bf M}}(X).
\end{equation*}
Define 
\begin{equation*}
\mathcal{I}/X := \{I \subseteq U \setminus X \mid {\sf p}_{{\bf M}}(I; X) = |I|\}, \ \ \ 
{\bf M}/X := (U \setminus X, \mathcal{I}/X).
\end{equation*}  
Then it is known (see, e.g., \cite[Proposition~3.1.6]{O11}) 
that ${\bf M}/X$ is a matroid.

\begin{lemma}[{See, e.g., \cite[Proposition~3.1.7]{O11}}] \label{lemma:contraction}
Assume that we are given a subset $X \subseteq U$ and a base $B$ of ${\bf M}|X$. 
Then for every subset $I \subseteq U \setminus X$, 
$I$ is an independent set {\normalfont (}resp., a base{\normalfont )} of ${\bf M}/ X$
if and only if 
 $I \cup B$ is an independent set {\normalfont (}resp., a base{\normalfont )} of ${\bf M}$. 
\end{lemma} 

\begin{lemma}[{See, e.g., \cite[Proposition~3.1.25]{O11}}] \label{lemma:minor}
For every pair of disjoint subsets $X,Y \subseteq U$, 
\begin{equation*}
({\bf M}/X) / Y = {\bf M}/ (X \cup Y), \ \ \ 
({\bf M}/X) | Y = ({\bf M}| (X \cup Y)) / X.
\end{equation*}  
\end{lemma}

The following lemma easily follows from 
Lemmas~\ref{lemma:contraction} and 
\ref{lemma:minor}. 
For completeness, we give its proof. 

\begin{lemma} \label{lemma:matroid_chain}
Assume that we are given a partition $\{U_1,U_2,\ldots,U_{\ell}\}$ of $U$.
Then there exists a base $B$ of ${\bf M}$ 
such that 
$B \cap U_{1,x}$ is 
a base of ${\bf M}| U_{1,x}$
for every integer $x \in \{1,2,\ldots,\ell\}$, 
where we define $U_{1,x} := \bigcup_{i=1}^x U_i$. 
\end{lemma}
\begin{proof}
We consider the algorithm described in Algorithm~\ref{alg:chain}.  
\begin{algorithm}[h]
Define ${\bf M}^{\prime}_0 := {\bf M}$. \\ 
\For{$i = 1,2,\ldots,\ell$}
{
    Find a base $B_i$ of ${\bf M}^{\prime}_{i-1}|U_i$. \\  
    Define ${\bf M}^{\prime}_i := {\bf M}^{\prime}_{i-1} / U_i$. 
}
Output $B_1 \cup B_2 \cup \cdots \cup B_{\ell}$ and halt. 
\caption{Algorithm for the proof of Lemma~\ref{lemma:matroid_chain}}
\label{alg:chain} 
\end{algorithm}

Assume that Algorithm~\ref{alg:chain} outputs $B$. 
We prove that $B$ satisfies the conditions in this lemma
by induction on $x$.
Since ${\bf M}^{\prime}_0|U_1 = {\bf M}|U_{1,1}$,  
$B$ satisfies the condition when $x = 1$. 

Assume that $B$ satisfies the condition when 
$x = j$ for some integer $j \in \{1,2,\ldots,\ell-1\}$. 
Then we consider the case in which $x = j+1$. 
The induction hypothesis implies that 
$B \cap U_{1,j}$ is a base of 
${\bf M}|U_{1,j}$. 
Line~3 of Algorithm~\ref{alg:chain} implies that 
$B \cap U_{j+1}= B_{j+1}$ is 
a base of ${\bf M}^{\prime}_{j}|U_{j + 1}$. 
Thus, since Lemma~\ref{lemma:minor} implies that 
\begin{equation*}
{\bf M}^{\prime}_j|U_{j+1} 
= ({\bf M} / U_{1,j}) | U_{j+1} 
= ({\bf M} | U_{1,j+1}) / U_{1,j}, \ \ \ 
{\bf M}|U_{1,j}  = 
({\bf M}|U_{1,j+1})|U_{1,j}, 
\end{equation*}
Lemma~\ref{lemma:contraction} implies that 
$B \cap U_{1,j+1}$ is 
a base of ${\bf M}|U_{1,j+1}$.
This completes the proof. 
\end{proof}

In our computation model, it is not difficult to see that for 
every symbol ${\rm S} \in \{{\rm D}, {\rm H}\}$ and 
every subset $F \subseteq E$, 
we can find a base of ${\bf M}_{\rm S}|F$ in polynomial time. 
Furthermore, we assume that we are given a symbol ${\rm S} \in \{{\rm D}, {\rm H}\}$,  
an independent set $I$ of ${\bf M}_{\rm S}$, and 
an element $e \in {\sf span}_{{\bf M}_{\rm S}}(I)$. 
Then we can compute ${\sf C}_{{\bf M}_{\rm S}}(e,I)$ in polynomial time as follows. 
Recall that for every element $f \in I + e$, 
$f \in {\sf C}_{{\bf M}_{\rm S}}(e,I)$ if and only if 
$I + e - f \in \mathcal{I}_{\rm S}$. 
Thus, we can compute ${\sf C}_{{\bf M}_{\rm S}}(e,I)$
in polynomial time
by determining whether $I + e - f \in \mathcal{I}_{\rm S}$
for all the elements $f \in I + e$. 

\section{Algorithm}

For each symbol ${\rm S} \in \{{\rm D}, {\rm H}\}$, 
we define 
${\sf C}_{\rm S}(\cdot, \cdot) := {\sf C}_{{\bf M}_{\rm S}}(\cdot, \cdot)$
and 
${\sf span}_{\rm S}(\cdot) := {\sf span}_{{\bf M}_{\rm S}}(\cdot)$.  

For each symbol ${\rm S} \in \{{\rm D}, {\rm H}\}$ and each 
subset $F \subseteq E$, 
we define ${\sf head}_{\rm S}(F)$ by 
\begin{equation*}
{\sf head}_{\rm S}(F) := 
\{e \in F \mid \mbox{$e \succsim_{\rm S} f$ 
for every element $f \in F$}\}.
\end{equation*}
Furthermore, for each subset $F \subseteq E$, 
we define ${\sf tail}_{\rm H}(F)$ by 
\begin{equation*}
{\sf tail}_{\rm H}(F) := 
\{e \in F \mid \mbox{$f \succsim_{\rm H} e$ 
for every element $f \in F$}\}.
\end{equation*}

For each independent set $I$ of ${\bf M}_{\rm H}$, 
we define ${\sf block}_{\rm H}(I)$ by 
\begin{equation*}
{\sf block}_{\rm H}(I) := 
\{e \in {\sf span}_{\rm H}(I) \mid 
\mbox{there exists an element $f \in {\sf D}_{\rm H}(e,I)$
such that $e \succsim_{\rm H} f$}\}. 
\end{equation*}
For every independent set $I$ of ${\bf M}_{\rm H}$, 
since we can compute ${\sf span}_{\rm H}(I)$ in polynomial time
and we can compute ${\sf D}_{\rm H}(e,I)$ in polynomial time 
for every element $e \in {\sf span}_{\rm H}(I)$, 
we can compute ${\sf block}_{\rm H}(I)$ in polynomial time. 

\subsection{First subroutine} 

Throughout this subsection, we assume that we are given a subset $F \subseteq E$. 

Define ${\rm Ch}_{\rm D}(F)$ as the output of Algorithm~\ref{alg:choice_D}. 
Since $H_i \neq \emptyset$ in the course of 
Algorithm~\ref{alg:choice_D}, 
the number of iterations of Lines~3 to 8 of Algorithm~\ref{alg:choice_D} 
is $O(|F|)$. 
Furthermore, Lemmas~\ref{lemma:contraction} and 
\ref{lemma:minor} imply that 
in Line~6, we can compute $P_i$ in polynomial time by finding a base of 
${\bf M}_{\rm D}|(\bigcup_{j=1}^{i-1}H_j)$. 
Thus, we can compute ${\rm Ch}_{\rm D}(F)$ in polynomial time. 

\begin{algorithm}[h]
Define $N_0 := F$, $I_0 := \emptyset$, and ${\bf N}_0 := {\bf M}_{\rm D}$. \\
Set $i := 0$. \\ 
\While{$N_i \neq \emptyset$}
{
    Set $i := i+1$. \\
    Define $H_i := {\sf head}_{\rm D}(N_{i-1})$ and $N_i := N_{i-1} \setminus H_i$. \\ 
    Define $P_i := \{e \in H_i \mid \mbox{$\{e\}$ is an independent set of ${\bf N}_{i-1}$}\}$ 
    and $I_i := I_{i-1} \cup P_i$. \\ 
    Define ${\bf N}_i := {\bf N}_{i-1} / H_i$. \\ 
}
Define $i_{\rm D} := i$. \\
Output $I_{i_{\rm D}}$ and halt.
\caption{${\rm Ch}_{\rm D}(F)$}  
\label{alg:choice_D} 
\end{algorithm}

In what follows, we define $I^{\ast} := {\rm Ch}_{\rm D}(F)$. 
For each integer $i \in \{1,2,\ldots,i_{\rm D}\}$, we define
\begin{equation*}
H_{1,i} := \bigcup_{j=1}^i H_j.
\end{equation*} 
Throughout this subsection, 
let $B^{\ast}$ be a base of ${\bf M}_{\rm D}|F$ such that 
$B^{\ast} \cap H_{1,i}$ is a base of ${\bf M}|H_{1,i}$
for every integer $i \in \{1,2,\ldots,i_{\rm D}\}$. 
Lemma~\ref{lemma:matroid_chain} guarantees the existence of $B^{\ast}$.

\begin{lemma} \label{lemma:b_subset}
$B^{\ast} \subseteq I^{\ast}$. 
That is, ${\rm Ch}_{\rm D}(F)$ contains a base of ${\bf M}_{\rm D}|F$. 
\end{lemma}
\begin{proof}
Let $i$ be an integer in $\{1,2,\ldots,i_{\rm D}\}$. 
In order to prove this lemma, it is sufficient to prove that 
$B^{\ast} \cap H_i \subseteq P_i$. 
Let $e$ be an element in $B^{\ast} \cap H_i$. 
If $i = 1$, then since $\{e\} \in \mathcal{I}_{\rm D}$ and 
${\bf N}_0 = {\bf M}_{\rm D}$, 
we have $e \in P_1$.
Thus, we assume that $i \ge 2$. 
Since $B^{\ast} \cap H_{1,i-1}$ is a base of ${\bf M}_{\rm D}|H_{1,i-1}$ and 
(I1) implies that 
$(B^{\ast} \cap H_{1,i-1}) + e \in \mathcal{I}_{\rm D}$, 
Lemma~\ref{lemma:contraction} 
implies that 
$\{e\}$ is an independent set of ${\bf M}_{\rm D} / H_{1,i-1}$. 
Thus, since Lemma~\ref{lemma:minor} implies that 
${\bf N}_{i-1} = {\bf M}_{\rm D} /H_{1,i-1}$, 
$\{e\} \in P_i$
for every element $e \in B^{\ast} \cap H_i$. 
This completes the proof. 
\end{proof} 

\begin{lemma} \label{lemma:base_choice_D}
If $I^{\ast} \in \mathcal{I}_{\rm D}$, then $I_i$ is a base of ${\bf M}_{\rm D}| H_{1,i}$
for every integer $i \in \{1,2,\ldots,i_{\rm D}\}$. 
\end{lemma}
\begin{proof}
Let $i$ be an integer in $\{1,2,\ldots,i_{\rm D}\}$. 
Then since $I^{\ast} \in \mathcal{I}_{\rm D}$, 
(I1) implies that $I_i = I^{\ast} \cap H_{1,i}$ is an 
independent set of ${\bf M}_{\rm D}|H_{1,i}$. 
Lemma~\ref{lemma:b_subset} implies that 
$B^{\ast} \cap H_{1,i} \subseteq I_i$. 
Thus, since $B^{\ast} \cap H_{1,i}$ is a base of ${\bf M}_{\rm D}|H_{1,i}$, 
$I_i$ is a base of ${\bf M}_{\rm D}|H_{1,i}$. 
This completes the proof. 
\end{proof}

Notice that for every pair of distinct integers $i,j \in \{1,2,\ldots,i_{\rm D}\}$
such that $i < j$ and every pair of elements $e \in H_i$ and 
$f \in H_j$, we have $e \succ_{\rm D} f$.  

\begin{lemma} \label{lemma:dominate}
Assume that $I^{\ast} \in \mathcal{I}_{\rm D}$ and 
we are given an element $e \in F \setminus I^{\ast}$.
Then $e \in {\sf span}_{\rm D}(I^{\ast})$ and 
$f \succ_{\rm D} e$ for every element $f \in {\sf D}_{\rm D}(e, I^{\ast})$. 
\end{lemma}
\begin{proof}
Since Lemma~\ref{lemma:base_choice_D}
implies that $I^{\ast}$ is a base of ${\bf M}_{\rm D}|F$ and $e \in F$, 
$e \in {\sf span}_{\rm D}(I^{\ast})$.
Assume that $e \in H_i$ for some integer $i \in \{1,2,\ldots,i_{\rm D}\}$.
Then since $e \notin I^{\ast}$, $e \in H_i \setminus P_i$.
This implies that $\{e\}$ is not an independent set of ${\bf N}_{i-1}$. 
Notice that since $\{e\} \in \mathcal{I}_{\rm D}$ and ${\bf N}_0 = {\bf M}_{\rm D}$,  
$i \ge 2$ holds. 
Lemma~\ref{lemma:minor} implies that 
${\bf N}_{i-1} = {\bf M}_{\rm D}/ H_{1,i-1}$. 
Thus, since 
Lemma~\ref{lemma:base_choice_D}
implies that 
$I_{i-1}$ is a base of ${\bf M}_{\rm D}|H_{1,i-1}$, 
Lemma~\ref{lemma:contraction} implies that 
$e \in {\sf span}_{\rm D}(I_{i-1})$. 
Furthermore, since $I_{i-1} \subseteq I^{\ast}$, 
Lemma~\ref{lemma:circuit_subset} implies that 
${\sf C}_{\rm D}(e,I^{\ast}) = {\sf C}_{\rm D}(e,I_{i-1})$. 
This implies that 
$f \succ_{\rm D} e$ for every element $f \in {\sf D}_{\rm D}(e, I^{\ast})$.
This completes the proof. 
\end{proof} 

\begin{lemma} \label{lemma:choice_D:upper_circuit} 
Assume that we are given an element $e \in F$, and 
there exists a circuit $C$ of ${\bf M}_{\rm D}$ such that 
$e \in C$, $C \subseteq F$, and 
$f \succ_{\rm D} e$ for every element $f \in C - e$. 
Then $e \notin {\rm Ch}_{\rm D}(F)$.
\end{lemma}
\begin{proof}
Assume that $e \in H_i$ for some integer $i \in \{1,2,\ldots,i_{\rm D}\}$. 
Notice that if $i = 1$, then $e \succsim_{\rm D} f$ for every element $f \in F$. 
Thus, since $\{e\} \in \mathcal{I}_{\rm D}$, we have $i \ge 2$. 

We prove that there exists a circuit $C^{\prime}$ of ${\bf M}_{\rm D}$ 
such that $e \in C^{\prime}$ and 
$C^{\prime} - e \subseteq B^{\ast} \cap H_{1,i-1}$. 
The existence of $C^{\prime}$ implies that 
$(B^{\ast} \cap H_{1,i-1}) + e \notin \mathcal{I}$.
Thus, since $B^{\ast} \cap H_{1,i-1}$ is a base of 
${\bf M}_{\rm D}|H_{1,i-1}$, Lemmas~\ref{lemma:contraction} and 
\ref{lemma:minor} imply that  
$\{e\}$ is not an independent set of 
${\bf N}_{i-1} = {\bf M}_{\rm D}/H_{1,i-1}$. 
Thus, $e \notin P_i$. This completes the proof. 

Define $B^{\prime} := B^{\ast} \cap H_{1,i-1}$.
Since $f \succ_{\rm D} e$
for every element $f \in C - e$, 
we have $C - e \subseteq H_{1,i-1}$. 
This implies that if $C - e \subseteq B^{\prime}$, then 
the proof is done. 
Thus, we assume that $(C - e) \setminus B^{\prime} \neq \emptyset$. 

In what follows, we prove that for every element $f \in (C-e) \setminus B^{\prime}$, 
there exists a circuit $C_f$ of ${\bf M}_{\rm D}$ such that 
$f \in C_f$ and $C_f - f \subseteq B^{\prime}$.  
If we can prove this, since $e \notin B^{\prime}$ follows from 
$e \in H_i$, 
Lemma~\ref{lemma:circuit_union} implies that 
there exists a circuit $C^{\prime}$ of ${\bf M}_{\rm D}$ such that 
\begin{equation*}
e \in C^{\prime} \subseteq \Big(C \cup 
\Big(\bigcup_{f \in (C -e) \setminus B^{\prime}}C_f\Big)\Big) 
\setminus \big((C - e) \setminus B^{\prime}\big).
\end{equation*} 
Since $C_f - f \subseteq B^{\prime}$ for every 
element $f \in (C - e) \setminus B^{\prime}$, 
$C^{\prime} - e\subseteq B^{\prime}$. 
This completes the proof. 

Let $f$ be an element in $(C - e) \setminus B^{\prime}$. 
Assume that $f \in H_x$ for some integer $x \in \{1,2,\ldots,i-1\}$. 
Since $f \notin B^{\prime}$, $f \notin B^{\ast} \cap H_{1,x}$. 
Thus, since $B^{\ast} \cap H_{1,x}$ is a base of ${\bf M}_{\rm D}|H_{1,x}$, 
$(B^{\ast} \cap H_{1,x}) + f \notin \mathcal{I}_{\rm D}$. 
Define $C_f := {\sf C}_{\rm D}(f, B^{\ast} \cap H_{1,x})$. 
Since $B^{\ast} \cap H_{1,x} \subseteq B^{\prime}$, 
$C_f - f \subseteq B^{\prime}$. 
This completes the proof. 
\end{proof}

\subsection{Second subroutine}

Throughout this subsection, we assume that we are given a subset $F \subseteq E$. 

Define ${\rm Ch}_{\rm H}(F)$ as the output 
of Algorithm~\ref{alg:choice_H}. 
Notice that in the course of 
Algorithm~\ref{alg:choice_H}, 
$I_i \in \mathcal{I}_{\rm H}$.  
Thus, ${\rm Ch}_{\rm H}(F) \in \mathcal{I}_{\rm H}$. 
Furthermore, since we can compute ${\sf C}_{\rm H}(e_i,I_{i-1})$ 
in polynomial time in Line~8, 
we can compute ${\rm Ch}_{\rm H}(F)$ in polynomial time. 

\begin{algorithm}[h]
Define $N_{0} := F$ and $I_{0} := \emptyset$. \\
Set $i := 0$. \\
\While{$N_{i}\neq \emptyset$}
{
    Set $i := i+1$. \\
    Define $e_{i}$ as an element in $N_{i-1}$. \\  
    Define $N_{i} := N_{i-1} - e_{i}$. \\
       \uIf{$e_i \in {\sf span}_{\rm H}(I_{i-1})$}
        {
            Define $I_{i} := (I_{i-1} + e_i) \setminus {\sf tail}_{\rm H}({\sf C}_{\rm H}(e_{i}, I_{i-1}))$. 
        }
        \Else
        {
            Define $I_{i} := I_{i-1} + e_i$.
        }
}
Define $i_{\rm H} := i$. \\
Output $I_{i_{\rm H}}$ and halt. 
\caption{${\rm Ch}_{\rm H}(F)$}  
\label{alg:choice_H} 
\end{algorithm}

\begin{lemma} \label{lemma:choice_H}
For every element $e \in F \setminus {\rm Ch}_{\rm H}(F)$, 
there exists a circuit $C$ of ${\bf M}_{\rm H}$ such that 
$C \subseteq F$, $e \in C$, 
$f \succsim_{\rm H} e$ for every element $f \in C$. 
\end{lemma}
\begin{proof}
Let $e$ be an element in $F \setminus {\rm Ch}_{\rm H}(F)$. 
Assume that 
$e = e_i$ for some integer $i \in \{1,2,\ldots,i_{\rm H}\}$.
Since $e \notin {\rm Ch}_{\rm H}(F)$, 
there exists an integer $j \in \{i,i+1,\ldots,i_{\rm H}\}$ such that 
$e \in {\sf tail}_{\rm H}({\sf C}_{\rm H}(e_j,I_{j-1}))$.
Define $C := {\sf C}_{\rm H}(e_j,I_{j-1})$. 
Then since $I_{j-1} + e_j \subseteq F$, 
we have $C \subseteq F$. 
Cleary, $e \in C$. 
Lastly, the definition of ${\sf tail}_{\rm H}(\cdot)$
implies that 
$f \succsim_{\rm H} e$ for every element $f \in C$. 
This completes the proof. 
\end{proof} 

\subsection{Description of algorithm} 

Our algorithm for the super-stable common independent 
set problem is Algorithm~\ref{alg:main}. 

\begin{algorithm}[h]
Define $I_0 := \emptyset$ and ${\sf R}_0 := \emptyset$. \\
Set $t := 0$. \\ 
\While{$I_t \neq {\rm Ch}_{\rm D}(E \setminus {\sf R}_t)$} 
{
    Set $t := t +1$. \\
    Define $J_{t,0} := \emptyset$ and ${\sf Q}_{t,0} := {\sf R}_{t-1}$. \\
    Set $i := 0$. \\
    \While{$J_{t,i} \neq {\rm Ch}_{\rm D}(E \setminus {\sf Q}_{t,i})$} 
    {
        Set $i := i+1$. \\
        Define $J_{t,i} := {\rm Ch}_{\rm H}({\rm Ch}_{\rm D}(E \setminus {\sf Q}_{t,i-1}))$. \\
        Define ${\sf Q}_{t,i} := {\sf Q}_{t,i-1} \cup 
        ({\rm Ch}_{\rm D}(E \setminus {\sf Q}_{t,i-1}) \setminus J_{t,i})$. 
    }
    Define $i_t := i$. \\ 
    Define $T_{t,0} := J_{t,i_t}$ and ${\sf S}_{t,0} := {\sf Q}_{t,i_t}$. \\
    Set $j := 0$. \\
   \While{${\sf S}_{t,j} \cap {\sf block}_{\rm H}(T_{t,j}) \neq \emptyset$}
   {
        Set $j := j+1$. \\
        Define $b_{t,j}$ as an element in ${\sf S}_{t,j-1} \cap {\sf block}_{\rm H}(T_{t,j-1})$.\\ 
        Define $T_{t,j} := T_{t,j-1} \setminus {\sf tail}_{\rm H}({\sf C}_{\rm H}(b_{t,j},T_{t,j-1}))$. \\
        Define ${\sf S}_{t,j} := {\sf S}_{t,j-1} \cup {\sf tail}_{\rm H}({\sf C}_{\rm H}(b_{t,j},T_{t,j-1}))$.  
   }
   Define $j_t := j$. \\
   Define $I_{t} := T_{t,j_t}$ and ${\sf R}_{t} := {\sf S}_{t,j_t}$. 
}
Define $k := t$. \\
\uIf{$I_k \notin \mathcal{I}_{\rm D}$}
{
Output {\bf null} and halt. 
}
\uElseIf{\normalfont there exists an element $e_{\rm R} \in {\sf R}_k$ 
such that $I_k + e_{\rm R} \in \mathcal{I}_{\rm H}$}
{
Output {\bf null} and halt. 
}
\Else
{
Output $I_k$ and halt. 
}
\caption{Algorithm for the super-stable common independent set problem}  
\label{alg:main} 
\end{algorithm}

\begin{lemma} \label{lemma1:iteration}
In each iteration of Lines~3 to 23, 
the number of iteration of Lines~7 to 11
is $O(|E|)$.
\end{lemma}
\begin{proof}
In order to prove this lemma, 
it is sufficient to prove that 
when $i = \ell$, if Algorithm~\ref{alg:main} does not go to Line~12, then 
${\sf Q}_{t,\ell-1} \subsetneq {\sf Q}_{t,\ell}$. 
Since ${\sf Q}_{t,\ell-1} \subseteq {\sf Q}_{t,\ell}$, 
it is sufficient to prove that 
${\sf Q}_{t,\ell-1} \neq {\sf Q}_{t,\ell}$.
In order to prove this by contradiction, we assume that 
${\sf Q}_{t,\ell-1} = {\sf Q}_{t,\ell}$. 
Since 
the definition of ${\rm Ch}_{\rm H}(\cdot)$ implies that 
$J_{t,\ell} \subseteq {\rm Ch}_{\rm D}(E \setminus {\sf Q}_{t,\ell-1})$, 
$J_{t,\ell} = {\rm Ch}_{\rm D}(E \setminus {\sf Q}_{t,\ell-1})$.
Recall that ${\sf Q}_{t,\ell-1} = {\sf Q}_{t,\ell}$.
This implies that 
$J_{t,\ell} = {\rm Ch}_{\rm D}(E \setminus {\sf Q}_{t,\ell})$.
This contradicts the fact that 
Algorithm~\ref{alg:main} does not go to Line~12
when $i = \ell$.
This completes the proof.  
\end{proof} 

\begin{lemma} \label{lemma2:iteration}
In each iteration of Lines~3 to 23, 
the number of iteration of Lines~15 to 20
is $O(|E|)$.
\end{lemma}
\begin{proof}
Notice that in the course of Algorithm~\ref{alg:main}, 
$T_{t,j} \cap {\sf S}_{t,j} = \emptyset$. 
In order to prove this lemma, 
since ${\sf S}_{t,j} \subseteq E$, 
it is sufficient to prove that 
\begin{equation*}
{\sf tail}_{\rm H}({\sf C}_{\rm H}(b_{t,j},T_{t,j-1})) \neq \{b_{t,j}\}.
\end{equation*} 
This follows from the fact that 
since $b_{t,j} \in {\sf block}_{\rm H}(T_{t,j-1})$, 
there exists 
an element $e \in {\sf D}_{\rm H}(b_{t,j},T_{t,j-1})$ 
such that 
$b_{t,j} \succsim_{\rm H} e$. 
This completes the proof. 
\end{proof} 

\begin{lemma} \label{lemma3:iteration} 
The number of iteration of Lines 3 to 23
is $O(|E|)$.
\end{lemma}
\begin{proof}
In order to prove this lemma, 
it is sufficient to prove that 
when $t = \ell$, 
if Algorithm~\ref{alg:main} does not go to Line~24, then 
${\sf R}_{\ell-1} \subsetneq {\sf R}_{\ell}$. 

First, we consider the case in which 
${\sf S}_{\ell,0} \cap {\sf block}_{\rm H}(T_{\ell,0}) = \emptyset$.
In this case, since 
$j_{\ell} = 0$,  
$I_{\ell} = J_{\ell,i_{\ell}}$ and ${\sf R}_{\ell} = {\sf Q}_{\ell,i_{\ell}}$.
Since $J_{t,\ell} = {\rm Ch}_{\rm D}(E \setminus {\sf Q}_{t,\ell})$, 
this contradicts the fact that 
Algorithm~\ref{alg:main} does not go to Line~24 when $t = \ell$. 

Next, we assume that 
${\sf S}_{\ell,0} \cap {\sf block}_{\rm H}(T_{\ell,0}) \neq \emptyset$.
Then 
$j_{\ell} \ge 1$ and 
\begin{equation*}
{\sf tail}_{\rm H}({\sf C}_{\rm H}(b_{\ell,1},T_{\ell,0})) \neq \{b_{\ell,1}\}, \ \ 
{\sf tail}_{\rm H}({\sf C}_{\rm H}(b_{\ell,1},T_{\ell,0})) 
\subseteq {\sf S}_{\ell,1} \setminus {\sf S}_{\ell,0}.
\end{equation*} 
Since $T_{\ell,0} \cap {\sf S}_{\ell,0} = \emptyset$, this implies that 
${\sf R}_{\ell-1} \subsetneq {\sf R}_{\ell}$. 
This completes the proof. 
\end{proof} 

Lemmas~\ref{lemma1:iteration}, \ref{lemma2:iteration}, and 
\ref{lemma3:iteration} imply that 
Algorithm~\ref{alg:main} is a polynomial-time algorithm. 

\section{Correctness} 

In this section, we prove the correctness of Algorithm~\ref{alg:main}. 

\begin{lemma} \label{lemma:stable}
If Algorithm~\ref{alg:main} outputs $I_k$, then 
$I_k$ is a super-stable common independent set of 
${\bf M}_{\rm D}$ and ${\bf M}_{\rm H}$.
\end{lemma}
\begin{proof}
Since Algorithm~\ref{alg:main} does not output 
{\bf null} in Line~26, 
$I_k \in \mathcal{I}_{\rm D}$. 
Furthermore, 
since ${\rm Ch}_{\rm H}(F) \in \mathcal{I}_{\rm H}$ for every subset $F \subseteq E$,
(I1) implies that $I_{k} \in \mathcal{I}_{\rm H}$.
These imply that $I_k$ is a common independent set of 
${\bf M}_{\rm D}$ and ${\bf M}_{\rm H}$. 
Thus, what remains is to prove that 
$I_k$ is super-stable. 

Assume that $I_k$ is not super-stable. 
Then in this case, there exists an element $e \in E \setminus I_k$ such that 
for each symbol ${\rm S} \in \{{\rm D}, {\rm H}\}$, one of the 
following conditions is satisfied.
\begin{description}
\setlength{\parskip}{0mm} 
\setlength{\itemsep}{1mm} 
\item[(B1)]
$e \notin {\sf span}_{\rm S}(I_k)$. 
\item[(B2)] 
$e \in {\sf span}_{\rm S}(I_k)$, and  
there exists an element $f \in {\sf D}_{\rm S}(e,I_k)$ 
such that $e \succsim_{\rm S} f$. 
\end{description}

First, we prove that $e \in {\sf R}_k$. 
Assume that 
$e \in E \setminus {\sf R}_k$. 
Since $e \notin I_k = {\rm Ch}_{\rm D}(E \setminus {\sf R}_k)$,
Lemma~\ref{lemma:dominate} implies that  
$e \in {\sf span}_{\rm D}(I_k)$ and 
$f \succ_{\rm D} e$ for every element $f \in {\sf D}_{\rm D}(e, I_k)$. 
However, this contradicts the fact that 
when ${\rm S} = {\rm D}$, 
one of (B1) and (B2) is satisfied.

Since Algorithm~\ref{alg:main} does not output {\bf null} in Line~28, 
when ${\rm S} = {\rm H}$, (B1) is not satisfied. 
Thus, (B2) is satisfied. 
That is, $e \in {\sf block}_{\rm H}(I_k)$. 
However, since $I_k = T_{k,j_k}$ and 
${\sf R}_k = {\sf S}_{k,j_k}$, 
${\sf R}_k \cap {\sf block}_{\rm H}(I_k) = \emptyset$. 
Since $e \in {\sf R}_k$, this is a contradiction. 
This completes the proof. 
\end{proof}

\begin{lemma} \label{lemma:R}
For every integer $t \in \{1,2,\ldots,k\}$ and 
every element $e \in J_{t,i_t} \cup {\sf Q}_{t,i_t}$, 
there exist integers $z \in \{1,2,\ldots,t\}$ 
and $i \in \{0,1,\ldots,i_z-1\}$ such that 
$e \in {\rm Ch}_{\rm D}(E \setminus {\sf Q}_{z,i})$.   
\end{lemma}
\begin{proof}
Assume that $e \in J_{t,i_t}$.
Notice that since $J_{t,0} = \emptyset$, $i_t \ge 1$. 
Since 
\begin{equation*}
e \in J_{t,i_t} = {\rm Ch}_{\rm H}({\rm Ch}_{\rm D}(E \setminus {\sf Q}_{t,i_t-1})) \subseteq 
{\rm Ch}_{\rm D}(E \setminus {\sf Q}_{t,i_t-1}), 
\end{equation*}
this lemma holds. Thus, in what follows, we assume that $e \in {\sf Q}_{t,i_t}$. 

First, we consider the case in which $t = 1$.
Then there exists an integer $\ell \in \{1,2,\ldots,i_1\}$ such that 
$e \in {\sf Q}_{1,\ell} \setminus {\sf Q}_{1,\ell-1}$. 
This implies that $e \in {\rm Ch}_{\rm D}(E \setminus {\sf Q}_{1,\ell-1})$. 

Let $x$ be an integer in $\{2,3,\ldots,k\}$. 
Then we assume that this lemma holds when $t = x-1$, and 
we consider the case in which $t = x$. 
If there exists an integer $\ell \in \{1,2,\ldots,i_x\}$ such that 
$e \in {\sf Q}_{x,\ell} \setminus {\sf Q}_{x,\ell-1}$. 
This implies that $e \in {\rm Ch}_{\rm D}(E \setminus {\sf Q}_{x,\ell-1})$. 
In what follows, we consider the case in which $e \in {\sf Q}_{x,0}$. 
In this case, $e \in {\sf R}_{x-1}$. 
Since 
\begin{equation*}
{\sf R}_{x-1} = {\sf S}_{x-1,j_{x-1}} \subseteq J_{x-1,i_{x-1}} \cup {\sf Q}_{x-1,i_{x-1}},
\end{equation*}
the hypothesis 
implies that there exist integers $y \in \{1,2,\ldots,x-1\}$ and $\ell \in \{0,1,\ldots,i_y-1\}$ such that 
$e \in {\rm Ch}_{\rm D}(E \setminus {\sf Q}_{y,\ell})$. 
This completes the proof. 
\end{proof} 

\begin{lemma} \label{lemma:reject} 
There does not exist a super-stable common independent set $M$ of 
${\bf M}_{\rm D}$ and ${\bf M}_{\rm H}$ such that 
$M \cap {\sf R}_k \neq \emptyset$. 
\end{lemma}
\begin{proof}
We call an element $e \in {\sf R}_k$ a {\em bad element} if 
there exists a super-stable common independent set $M$ of 
${\bf M}_{\rm D}$ and ${\bf M}_{\rm H}$ such that 
$e \in M$. 
In order to prove this lemma, it is sufficient to prove that 
there does not exist a bad element in ${\sf R}_k$. 

Assume that there exist a bad element in ${\sf R}_k$. 
Let $\Delta$ be the set of integer $t \in \{1,2,\ldots,k\}$ such that 
there exists a bad element in ${\sf R}_t \setminus {\sf R}_{t-1}$. 
We denote by $z$ the minimum integer in $\Delta$. 

First, we consider the case in which 
there exists a bad element in ${\sf Q}_{z,i_z} \setminus {\sf R}_{z-1}$. 
We denote by $\Delta^{\prime}$ the set of integers $i \in \{1,2,\ldots,i_z\}$
such that there exists a bad element in 
${\sf Q}_{z,i} \setminus {\sf Q}_{z,i-1}$. 
Define $x$ as the minimum integer in $\Delta^{\prime}$.
Let $e$ be a bad element in ${\sf Q}_{z,x} \setminus {\sf Q}_{z,x-1}$.
The definition of a bad element in ${\sf R}_k$ implies that 
there exists a super-stable common independent set $M$ 
of ${\bf M}_{\rm D}$ and ${\bf M}_{\rm H}$ such that 
$e \in M$.
Since $e \in {\sf Q}_{z,x} \setminus {\sf Q}_{z,x-1}$, 
$e \in {\rm Ch}_{\rm D}(E \setminus {\sf Q}_{z,x-1}) \setminus J_{z,x}$. 
Define $X := {\rm Ch}_{\rm D}(E \setminus {\sf Q}_{z,x-1})$. 
Then $e \in X \setminus {\rm Ch}_{\rm H}(X)$. 
This and Lemma~\ref{lemma:choice_H} imply that 
there exists a circuit $C$ of ${\bf M}_{\rm H}$ such that 
$C \subseteq {\rm Ch}_{\rm D}(E \setminus {\sf Q}_{z,x-1})$, 
$e \in C$, and $f \succsim_{\rm H} e$ for every 
element $f \in C$. 

\begin{claim} \label{claim1:lemma:reject}
For every element $f \in C \setminus M$, 
there exists a circuit $C_f$ of ${\bf M}_{\rm H}$ such that 
\begin{equation*}
f \in C_f, \ \ e \notin C_f, \ \ C_f - f \subseteq M.
\end{equation*} 
\end{claim}
\begin{proof}
Let $f$ be an element in $C \setminus M$. 
Since $e \in M$, $f \neq e$. 

First, we prove that if $f \in {\sf span}_{\rm D}(M)$, then 
\begin{equation} \label{eq1:claim1:lemma:reject}
{\sf D}_{\rm D}(f,M) \subseteq E \setminus {\sf Q}_{z,x-1}.
\end{equation} 
Assume that 
\eqref{eq1:claim1:lemma:reject} does not hold, that is, 
${\sf D}_{\rm D}(f,M) \cap {\sf Q}_{z,x-1} \neq \emptyset$. 
Let $g$ be an element in ${\sf D}_{\rm D}(f,M) \cap {\sf Q}_{z,x-1}$. 
Since $g \in M \cap {\sf Q}_{z,x-1}$, 
$g$ is a bad element in ${\sf R}_k$. 
This contradicts the definition of $x$. 

Next, we prove that if 
\begin{equation} \label{eq2:claim1:lemma:reject}
f \notin {\bf dom}_{\rm D}(M),
\end{equation}
then the proof is done. 
Assume that \eqref{eq2:claim1:lemma:reject} holds.
Then since $M$ is super-stable, 
$f \in {\bf dom}_{\rm H}(M)$, i.e.,  
$f \in {\sf span}_{\rm H}(M)$ and 
$g \succ_{\rm H} f$ for every element $g \in {\sf D}_{\rm H}(f,M)$. 
Since $f \succsim_{\rm H} e$ and $f \neq e$, 
we have $e \notin {\sf C}_{\rm H}(f,M)$. 
This implies that ${\sf C}_{\rm H}(f,M)$ 
satisfies the conditions in this claim. 

Here we prove \eqref{eq2:claim1:lemma:reject}.
Assume that \eqref{eq2:claim1:lemma:reject} does not hold.
Then 
$f \in {\sf span}_{\rm D}(M)$ and 
$g \succ_{\rm D} f$ for every element $g \in {\sf D}_{\rm D}(f,M)$.
Since $f \in {\rm Ch}_{\rm D}(E \setminus {\sf Q}_{z,x-1})$, 
\eqref{eq1:claim1:lemma:reject} implies that 
${\sf C}_{\rm D}(f,M) \subseteq E\setminus {\sf Q}_{z,x-1}$. 
Thus,
Lemma~\ref{lemma:choice_D:upper_circuit} implies that 
$f \notin {\rm Ch}_{\rm D}(E \setminus {\sf Q}_{z,x-1})$.
This contradicts the fact that 
$f \in {\rm Ch}_{\rm D}(E \setminus {\sf Q}_{z,x-1})$.
This completes the proof. 
\end{proof} 

Lemma~\ref{lemma:circuit_union} and Claim~\ref{claim1:lemma:reject} imply that 
there exists a circuit $C^{\prime}$ of ${\bf M}_{\rm H}$ such that 
\begin{equation*}
C^{\prime} \subseteq \Big(C \cup \Big(\bigcup_{f \in C \setminus M}C_f\Big)\Big)
\setminus (C \setminus M) \subseteq M.
\end{equation*}
This contradicts the fact that $M \in \mathcal{I}_{\rm H}$. 

Next, we consider the case in which 
there does not exist a bad element in ${\sf Q}_{z,i_z} \setminus {\sf R}_{z-1}$. 
In this case, there exists a bad element in ${\sf S}_{z,j_z} \setminus {\sf Q}_{z,i_z}$. 
We denote by $\Delta^{\prime}$ the set of integers $j \in \{1,2,\ldots,j_z\}$
such that there exists a bad element in 
${\sf S}_{z,j} \setminus {\sf S}_{z,j-1}$. 
Define $x$ as the minimum integer in $\Delta^{\prime}$.
Let $e$ be a bad element in ${\sf S}_{z,x} \setminus {\sf S}_{z,x-1}$.
Then the definition of a bad element implies that there exists 
a super-stable common independent set $M$ 
of ${\bf M}_{\rm D}$ and ${\bf M}_{\rm H}$ such that 
$e \in M$.
Since $e \in {\sf S}_{z,x} \setminus {\sf S}_{z,x-1}$, 
\begin{equation*}
e \in {\sf tail}_{\rm H}({\sf C}_{\rm H}(b_{z,x}, T_{z,x-1})).
\end{equation*} 
Define $C := {\sf C}_{\rm H}(b_{z,x}, T_{z,x-1})$. Then 
$e \in C$ and 
$f \succsim_{\rm H} e$ for every element $f \in C$. 

\begin{claim} \label{claim2:lemma:reject}
For every element $f \in C \setminus M$, 
there exists a circuit $C_f$ of ${\bf M}_{\rm H}$ such that 
\begin{equation*}
f \in C_f, \ \ e \notin C_f, \ \ C_f - f \subseteq M.
\end{equation*} 
\end{claim}
\begin{proof}
Let $f$ be an element in $C \setminus M$. 
Since $e \in M$, $f \neq e$. 

First, we prove that if $f \in {\sf span}_{\rm D}(M)$, then 
\begin{equation} \label{eq1:claim2:lemma:reject}
{\sf D}_{\rm D}(f,M) \subseteq E \setminus {\sf S}_{z,x-1}.
\end{equation} 
Assume that 
\eqref{eq1:claim2:lemma:reject} does not hold, that is, 
${\sf D}_{\rm D}(f,M) \cap {\sf S}_{z,x-1} \neq \emptyset$. 
Let $g$ be an element in ${\sf D}_{\rm D}(f,M) \cap {\sf S}_{z,x-1}$. 
Since $g \in M \cap {\sf S}_{z,x-1}$, 
$g$ is a bad element in ${\sf R}_k$. 
This contradicts the definition of $x$. 

Next, we prove that if 
\begin{equation} \label{eq2:claim2:lemma:reject}
f \notin {\bf dom}_{\rm D}(M),
\end{equation}
then the proof is done. 
Assume that \eqref{eq2:claim2:lemma:reject} holds.
Then since $M$ is super-stable, 
$f \in {\bf dom}_{\rm H}(M)$, i.e.,  
$f \in {\sf span}_{\rm H}(M)$ and 
$g \succ_{\rm H} f$ for every element $g \in {\sf D}_{\rm H}(f,M)$. 
Since $f \succsim_{\rm H} e$ and $f \neq e$, 
we have $e \notin {\sf C}_{\rm H}(f,M)$. 
This implies that ${\sf C}_{\rm H}(f,M)$ 
satisfies the conditions in this claim. 

What remains is to prove 
\eqref{eq2:claim2:lemma:reject}.
In what follows, we assume that 
\eqref{eq2:claim2:lemma:reject} does not hold.
This implies that 
$f \in {\sf span}_{\rm D}(M)$ and 
$g \succ_{\rm D} f$ for every element $g \in {\sf D}_{\rm D}(f,M)$.

Recall that $f \in T_{z,x-1} + b_{z,x}$. 
Thus, since $b_{z,x} \in {\sf S}_{z,x-1}$, 
\begin{equation*}
f \in T_{z,x-1} + b_{z,x} \subseteq J_{z,i_z} \cup {\sf Q}_{z,i_z}. 
\end{equation*}
Thus, 
Lemma~\ref{lemma:R}
implies that there exist integers $y \in \{1,2,\ldots, z\}$ and 
$i \in \{0,1,\ldots,i_y-1\}$ such that 
$f \in {\rm Ch}_{\rm D}(E \setminus {\sf Q}_{y,i})$. 
Since $f \in {\rm Ch}_{\rm D}(E \setminus {\sf Q}_{y,i})$, 
$f \in E \setminus {\sf Q}_{y,i}$. 
Since ${\sf Q}_{y,i} \subseteq {\sf S}_{z,x-1}$, 
\eqref{eq1:claim2:lemma:reject} implies that
${\sf D}_{\rm D}(f,M) \subseteq E \setminus {\sf Q}_{y,i}$.
Thus, 
${\sf C}_{\rm D}(f,M) \subseteq E \setminus {\sf Q}_{y,i}$. 
This and Lemma~\ref{lemma:choice_D:upper_circuit} 
imply that 
$f \notin {\rm Ch}_{\rm D}(E \setminus {\sf Q}_{y,i})$. 
However, this contradicts the fact that 
$f \in {\rm Ch}_{\rm D}(E \setminus {\sf Q}_{y,i})$.
This completes the proof.
\end{proof} 

Lemma~\ref{lemma:circuit_union} and Claim~\ref{claim2:lemma:reject} imply that 
there exists a circuit $C^{\prime}$ of ${\bf M}_{\rm H}$ such that 
\begin{equation*}
C^{\prime} \subseteq \Big(C \cup \Big(\bigcup_{f \in C \setminus M}C_f\Big)\Big)
\setminus (C \setminus M) \subseteq M.
\end{equation*}
This contradicts the fact that $M \in \mathcal{I}_{\rm H}$. 
This completes the proof. 
\end{proof} 

\begin{lemma} \label{lemma:null} 
If Algorithm~\ref{alg:main} outputs {\bf null}, 
then there does not exist a super-stable common independent set of 
${\bf M}_{\rm D}$ and ${\bf M}_{\rm H}$. 
\end{lemma}
\begin{proof} 
Assume that Algorithm~\ref{alg:main} outputs {\bf null}.
Furthermore, we assume that there exists a super-stable common independent set $M$ of 
${\bf M}_{\rm D}$ and ${\bf M}_{\rm H}$. 

\begin{claim} \label{claim1:lemma:null} 
There exists an independent set $L$ of ${\bf M}_{\rm H}$ 
such that $L \subseteq I_k \cup {\sf R}_k$ and $|L| > |M|$. 
\end{claim} 
\begin{proof} 
Since 
Lemma~\ref{lemma:reject} implies that 
$M \subseteq E \setminus {\sf R}_k$, 
$M$ is an independent set of 
${\bf M}_{\rm D}|(E \setminus {\sf R}_k)$. 

First, we assume that 
Algorithm~\ref{alg:main} outputs {\bf null} in Line~26. 
In this case, we define $L := I_k$.
The definition of ${\rm Ch}_{\rm H}(\cdot)$ implies that 
$L \in \mathcal{I}_{\rm H}$. 
Lemma~\ref{lemma:b_subset} implies that 
$I_k$ contains a base of ${\bf M}_{\rm D}|(E \setminus {\sf R}_k)$. 
Thus, since $I_k \notin \mathcal{I}_{\rm D}$, 
$|I_k| > {\sf r}_{{\bf M}_{\rm D}}(E \setminus {\sf R}_k)$. 
This implies that 
\begin{equation*}
|L| = |I_k| > {\sf r}_{{\rm M}_{\rm D}}(E \setminus {\sf R}_k) \ge |M|.
\end{equation*} 

Next, we assume that 
Algorithm~\ref{alg:main} outputs {\bf null} in Line~28. 
In this case, we define $L := I_k + e_{\rm R}$. 
The definition of $e_{\rm R}$ implies that 
$L \in \mathcal{I}_{\rm H}$. 
Since $e_{\rm R} \in {\sf R}_k$, 
$L \subseteq I_k \cup {\sf R}_k$. 
Lemma~\ref{lemma:b_subset} implies that 
$I_k$ contains a base of ${\bf M}_{\rm D}|(E \setminus {\sf R}_k)$. 
This implies that $|I_k| \ge |M|$. 
Thus, since $e_{\rm R} \notin I_k$ follows from 
the fact that $I_k \subseteq E \setminus {\sf R}_k$, we have 
\begin{equation*}
|L| = |I_k + e_{\rm R}| = |I_k| + 1 > |I_k| \ge |M|.
\end{equation*} 
This completes the proof. 
\end{proof} 

In what follows, let $L$ be an independent set of ${\bf M}_{\rm H}$ 
such that $L \subseteq I_k \cup {\sf R}_k$ and $|L| > |M|$. 

\begin{claim} \label{claim2:lemma:null} 
$e \in {\sf span}_{\rm H}(M)$ 
for every element $e \in L \setminus M$.
\end{claim}
\begin{proof}
Let $e$ be an element in $L \setminus M$. 
Since 
\begin{equation*}
e \in L \subseteq I_k \cup {\sf R}_k  
= T_{k,j_k} \cup {\sf S}_{k,j_k} 
= J_{k,i_k} \cup {\sf Q}_{k,i_k}, 
\end{equation*}
it follows from Lemma~\ref{lemma:R} 
that 
there exist integers $t \in \{1,2,\ldots, k\}$ and 
$i \in \{0,1,\ldots,i_t-1\}$ such that 
$e \in {\rm Ch}_{\rm D}(E \setminus {\sf Q}_{t,i})$.
This implies that $e \in E \setminus {\sf Q}_{t,i}$. 

Assume that $e \in {\bf dom}_{\rm D}(M)$, 
i.e.,  $e \in {\sf span}_{\rm D}(M)$ and 
$f \succ_{\rm D} e$ for every 
element $f \in {\sf D}_{\rm D}(e,M)$. 
Since $M \subseteq E \setminus {\sf R}_k$, 
${\sf D}_{\rm D}(e,M) \subseteq E \setminus {\sf R}_k$. 
Since ${\sf Q}_{t,i} \subseteq {\sf R}_k$, 
${\sf C}_{\rm D}(e,M) \subseteq E \setminus {\sf Q}_{t,i}$.
This and Lemma~\ref{lemma:reject} imply that 
$e \notin  {\rm Ch}_{\rm D}(E \setminus {\sf Q}_{t,i})$.
This is a contradiction. 
Thus, $e \notin {\bf dom}_{\rm D}(M)$. 

Since $M$ is super-stable, 
$e \in {\bf dom}_{\rm H}(M)$.
This implies that 
$e \in {\sf span}_{\rm H}(M)$. 
\end{proof} 

In what follows, we prove that there exists a circuit $C$ of ${\bf M}_{\rm H}$ 
such that $C \subseteq L$. 
This contradicts the fact that $L \in \mathcal{I}_{\rm H}$. 
Thus, this completes the proof. 
Here we consider Algorithm~\ref{alg:circuit}. 
Notice that Claim~\ref{claim2:lemma:null} implies that 
$C_1^e$ is well-defined for every element $e \in L \setminus M$. 

\begin{algorithm}[h]
Define $C_1^e := {\sf C}_{\rm H}(e,M)$ for each element $e \in L \setminus M$. \\
Define $\mathcal{C}_1 := \{C_1^e \mid e \in L \setminus M\}$, 
$A_1 := L \setminus M$, and $V_1 := M \setminus L$. \\
Set $i := 1$. \\ 
\While{$C \not\subseteq L$ holds for every circuit $C \in \mathcal{C}_{i}$}
{
        Define $v_i$ as an element in $V_{i}$ which is contained in 
        at least two distinct circuits in $\mathcal{C}_{i}$. \\
        Let $a_i$ be an element in $A_i$ such that $v_i \in C_{i}^{a_i}$.  \\
        \For{each element $e \in A_i - a_i$}
        {
            \uIf{$v_i \in C_{i}^e$}
            {
                Define $C_{i+1}^e$ as a circuit $C$ of ${\bf M}_{\rm H}$ such that 
                $e \in C \subseteq (C_{i}^e \cup C_{i}^{a_i}) - v_i$.
            }
            \Else
            {
                Define $C_{i+1}^e := C_{i}^e$. 
            } 
        }
        Define $A_{i+1} := A_{i} - a_i$ and $V_{i+1} := V_{i} - v_i$. \\
        Define $\mathcal{C}_{i+1} := \{C_{i+1}^e \mid e \in A_{i+1}\}$. \\
        Set $i := i+1$. \\
}
Output a circuit $C \in \mathcal{C}_i$ such that $C \subseteq L$ and halt.
\caption{Algorithm for the proof of Lemma~\ref{lemma:null}}  
\label{alg:circuit} 
\end{algorithm}

\begin{claim} \label{claim5:lemma:null}
In the course of Algorithm~\ref{alg:circuit}, 
the following statements hold. 
\begin{description}
\setlength{\parskip}{0mm} 
\setlength{\itemsep}{1mm} 
\item[(D1)]
For every element $e \in A_i$, 
$C_i^e \cap (M \setminus L) \subseteq V_i$. 
\item[(D2)]
$|A_i| > |V_i|$. 
\item[(D3)]
For every element $e \in A_i$, 
$C_i^e \cap A_i = \{e\}$.
\end{description}
Furthermore, if 
$C \not\subseteq L$ holds for every 
circuit $C \in \mathcal{C}_{i}$, 
then the following statements hold. 
\begin{description}
\setlength{\parskip}{0mm} 
\setlength{\itemsep}{1mm} 
\item[(D4)]
There exists an element in $V_i$ which is 
contained in at least two distinct circuits in $\mathcal{C}_{i}$. 
\item[(D5)] 
In Line~9, there exists a circuit $C$ of ${\bf M}_{\rm H}$ 
such that $e \in C \subseteq (C_{i}^e \cup C_{i}^{a_i}) - v_i$. 
\end{description}
\end{claim}
\begin{proof}
First, we consider the case in which $i = 1$. 
Since $V_1 = M \setminus L$, 
(D1) holds. 
Since $|L \setminus M| > |M \setminus L|$, 
(D2) holds.
Since $A_1 = L \setminus M$ and 
$C_1^e - e \subseteq M$ for every element $e \in L \setminus M$,
(D3) holds. 
Notice that (D3) implies that 
$C_1^e \neq C_1^{\overline{e}}$ holds for every pair of distinct 
elements $e,\overline{e} \in A_1$. 
Thus, since every circuit in $\mathcal{C}_1$ contains at least one element in $M \setminus L$, 
(D4) follows from (D1) and (D2). 
For every element $e \in A_1 - a_1$, 
since $e \neq a_1$, (D3) implies that 
$C_1^e \neq C_1^{a_1}$.
Thus, 
Lemma~\ref{lemma:strong_elimination} implies that 
(D5) holds. 

Next, we assume that the statements in this claim hold when $i = x$ for some 
positive integer $x$, and 
we prove that 
the statements in this claim hold when $i = x + 1$. 
For every element $e \in A_{x+1}$, 
$C_{x}^e \cap (M \setminus L), C_{x}^{a_{x}} \cap (M \setminus L) \subseteq V_{x}$ 
and $v_{x} \notin C_{x+1}^e$. 
Thus, since $V_{x + 1} = V_{x} - v_{x}$, 
(D1) holds.
Since 
$|A_{x+1}| = |A_{x}| - 1$
and 
$|V_{x+1}| = |V_{x}| - 1$,
(D2) follows from $|A_{x}| > |V_{x}|$. 
For every element $e \in A_{x+1}$, 
$C_{x+1}^e \cap A_{x} \subseteq \{e,a_{x}\}$.
Thus, since $A_{x+1} = A_{x}-a_{x}$, 
(D3) holds. 
Notice that (D3) implies that 
$C_{x+1}^{e} \neq C_{x+1}^{\overline{e}}$ holds 
for every pair of distinct elements $e,\overline{e} \in A_{x+1}$. 
Thus, since every circuit in $\mathcal{C}_{x+1}$ 
contains at least one element in $M \setminus L$, 
(D4) follows from (D1) and (D2). 
For every element $e \in A_{x+1} - a_{x+1}$, 
since $e \neq a_{x+1}$, 
(D3) implies that 
$C_{x+1}^e \neq C_{x+1}^{a_{x+1}}$. 
Thus, Lemma~\ref{lemma:strong_elimination} implies that 
(D5) holds. 
This completes the proof. 
\end{proof}

In the course of Algorithm~\ref{alg:circuit}, 
$|A_{i+1}| < |A_i|$. 
Assume that in the course of Algorithm~\ref{alg:circuit},
$|A_i| = 1$ holds for some positive integer $i$.  
Then (D2) of Claim~\ref{claim5:lemma:null} implies that 
$V_i = \emptyset$. 
Let $C$ be the unique circuit in $\mathcal{C}_i$. 
Then (D1) of Claim~\ref{claim5:lemma:null} implies that 
$C \cap (M \setminus L) = \emptyset$. 
Thus, since $C \subseteq L \cup M$, 
$C \subseteq L$. 
This implies that 
Algorithm~\ref{alg:circuit} outputs 
a circuit $C$ of ${\bf M}_{\rm H}$ such that $C \subseteq L$ and halts.
This completes the proof.
\end{proof} 

\begin{theorem}
Algorithm~\ref{alg:main} can solve 
the super-stable common independent set problem. 
\end{theorem}
\begin{proof} 
This theorem follows from 
Lemmas~\ref{lemma:stable} and 
\ref{lemma:null}. 
\end{proof}

\appendix

\section{The Student-Project Allocation Problem with Ties} 
\label{appendix:spa} 

Here we show that the problem of finding a super-stable matching in 
the student-project allocation problem with ties, which 
was considered by Olaosebikan and Manlove~\cite{OM20+}, is a special case of 
the super-stable common independent set problem. 

The {\em student-project allocation problem with ties} is defined as follows. 
We are given a finite set $S = \{1,2,\ldots,n\}$ of {\em students}, a finite set $P$ of 
{\em projects}, and a finite set $L = \{1,2,\ldots,m\}$ of {\em lecturers}. 
The set $P$ is partitioned into  
$P_1, P_2, \ldots, P_m$. 
For each integer $\ell \in \{1,2,\ldots,m\}$, 
$P_{\ell}$ represents the set of projects which 
the lecturer $\ell$ offers. 
Furthermore, we are given a subset $\mathcal{A} \subseteq S \times P$. 
For each subset $M \subseteq \mathcal{A}$ and 
each student $s \in S$ (resp., project $p \in P$), we define 
$M(s)$ (resp., $M(p)$) 
the set of ordered pairs $(s^{\prime},p^{\prime}) \in M$ 
such that $s^{\prime} = s$ (resp., $p^{\prime} = p$). 
For each subset $M \subseteq \mathcal{A}$ and 
each lecturer $\ell \in L$, we define 
$M(\ell) := \bigcup_{p \in P_{\ell}}M(p)$. 
For each student $s \in S$, we are given a 
transitive and complete binary relation $\succsim_s$ on $\mathcal{A}(s)$. 
For each lecturer $\ell \in L$, we are given a 
transitive and complete binary relation $\succsim_{\ell}$ on $S$. 
Furthermore, we are given capacity functions 
$c_P \colon P \to \mathbb{Z}_+$ and 
$c_L \colon L \to \mathbb{Z}_+$, where 
$\mathbb{Z}_+$ denotes the set of non-negative integers. 

A subset $M \subseteq \mathcal{A}$ is called a {\em matching} if 
$|M(s)| \le 1$ for every student $s \in S$, 
$|M(p)| \le c_P(p)$ for every project $p \in P$, and 
$|M(\ell)| \le c_L(\ell)$ for every lecturer $\ell \in L$.
For each matching $M$ and each ordered pair $(s,p) \in \mathcal{A} \setminus M$, 
we say that $(s,p)$ {\em blocks} $M$ if the following conditions are satisfied.
\begin{itemize}
\setlength{\parskip}{0mm} 
\setlength{\itemsep}{1mm} 
\item
For every student $s \in S$, one of the following conditions is satisfied.
\begin{enumerate}
\setlength{\parskip}{0mm} 
\setlength{\itemsep}{1mm} 
\item
$M(s) = \emptyset$.
\item
$|M(s)| = 1$, and $(s,p) \succsim_s (s,q)$ for the ordered pair $(s,q) \in M(s)$. 
\end{enumerate}
\item
Let $\ell$ be the lecturer in $L$ such that 
$p \in P_{\ell}$. 
Then one of the following conditions is satisfied. 
\begin{enumerate}
\setlength{\parskip}{0mm} 
\setlength{\itemsep}{1mm} 
\item
$|M(p)| < c_P(p)$
and $|M(\ell)| < c_L(\ell)$. 
\item
$|M(p)| = c_P(p)$ and 
$s \succsim_{\ell} t$ for some ordered pair $(t,p) \in M(p)$. 
\item
$|M(p)| < c_P(p)$, 
$|M(\ell)| = c_L(\ell)$, and 
$s \succsim_{\ell} t$ for some ordered pair $(t,q) \in M(\ell)$. 
\end{enumerate}
\end{itemize}
A matching $M$ is said to be {\em super-stable} if 
there does not exist an ordered pair in $\mathcal{A} \setminus M$
which blocks $M$. 

We can reduce the problem of finding a super-stable matching 
in the student-project allocation problem with ties
to the super-stable common independent set problem as follows. 
Define $E := \mathcal{A}$. 
Define $\mathcal{I}_{\rm D}$ as the family of subsets 
$F \subseteq \mathcal{A}$ such that 
$|F(s)| \le 1$ for every student $s \in S$.  
Define $\mathcal{I}_{\rm H}$ as the 
family of subsets $F \subseteq \mathcal{A}$ 
such that $|F(p)| \le c_P(p)$ for every 
project $p \in P$ and 
$|F(\ell)| \le c_L(\ell)$
for every lecturer $\ell \in L$. 
Define $\succsim_{\rm D}$ as follows. 
\begin{itemize}
\setlength{\parskip}{0mm} 
\setlength{\itemsep}{1mm} 
\item
For each student $s \in S$ and each pair of 
ordered pairs $(s,p), (s,q) \in \mathcal{A}(s)$,  
$(s,p) \succsim_{\rm D} (s,q)$ if and only if 
$(s,p) \succsim_s (s,q)$. 
\item
For each pair of distinct students $s_1,s_2 \in S$ such that 
$s_1 < s_2$ and each pair of ordered pairs $e \in \mathcal{A}(s_1)$
and $f \in \mathcal{A}(s_2)$,  
$e \succ_{\rm D} f$.
\end{itemize}
Define $\succsim_{\rm H}$ as follows. 
\begin{itemize}
\setlength{\parskip}{0mm} 
\setlength{\itemsep}{1mm} 
\item
For each lecturer $\ell \in L$ and each pair of ordered pairs 
$(s,p), (t,q) \in \mathcal{A}(\ell)$,  
$(s,p) \succsim_{\rm H} (t,q)$ if and only if 
$s \succsim_{\ell} t$. 
\item
For each pair of distinct lecturers $\ell_1,\ell_2 \in L$ such that 
$\ell_1 < \ell_2$ and each pair of ordered pairs $e \in \mathcal{A}(\ell_1)$
and $f \in \mathcal{A}(\ell_2)$,  
$e \succ_{\rm H} f$. 
\end{itemize}
It is not difficult to see that 
for every subset $M \subseteq \mathcal{A}$, 
$M$ is a super-stable matching if and only if 
$M$ is a super-stable common independent set of 
${\bf M}_{\rm D}$ and 
${\bf M}_{\rm H}$. 

\subsubsection*{Acknowledgements} 

This work was supported by JST, PRESTO Grant Number JPMJPR1753, Japan.

\end{document}